\documentclass{article}

\usepackage[final]{neurips_2019}

\usepackage[utf8]{inputenc}
\usepackage[T1]{fontenc}
\usepackage{hyperref}
\usepackage{url}
\usepackage{booktabs}
\usepackage{amsfonts}
\usepackage{nicefrac}
\usepackage{microtype}
\usepackage{graphicx}
\usepackage{xcolor}
\usepackage{lipsum}
\usepackage{float}
\usepackage[square,sort,comma,numbers]{natbib}
\usepackage{subcaption}
\usepackage{mathrsfs}
\usepackage{amsmath,amsthm,verbatim,amssymb,amsfonts,amscd, graphicx, float}
\usepackage{thmtools}
\usepackage{thm-restate}

\usepackage{hyperref}
\usepackage[noabbrev]{cleveref}
\Crefname{subsection}{subsection}{subsections}

\newtheorem{theorem}{Theorem}[section]

\theoremstyle{definition}
\newtheorem{definition}{Definition}[section]

\newtheorem{assumption}{Assumption}

\newcommand{\Ppoly}{\textbf{P}/poly}
\newcommand{\NCo}{\textbf{NC}^{1}}
\newcommand{\iO}{i\mathcal{O}}
\newcommand{\ckt}{\mathcal{C}}
\newcommand{\Zp}{\mathbb{Z}_p}

\newcommand{\zo}{\{0,1\}}

\newcommand{\KeyG}{\textbf{KeyGen}}
\newcommand{\Encrypt}{\textbf{Enc}}
\newcommand{\Dec}{\textbf{Dec}}
\newcommand{\Eval}{\textbf{Eval}}
\newcommand{\getsr}{\mathrel{\mathpalette\rgetscmd\relax}}
\newcommand{\rgetscmd}{\ooalign{$\leftarrow$\cr
    \hidewidth\raisebox{1.2\height}{\scalebox{0.5}{\ \rm \$}}\hidewidth\cr}}
\newcommand{\tor}{\mathrel{\mathpalette\rtocmd\relax}}
\newcommand{\rtocmd}{\ooalign{$\rightarrow$\cr
    \hidewidth\raisebox{1.2\height}{\scalebox{0.5}{\ \rm \$}}\hidewidth\cr}}

\newcommand{\JGen}{\textbf{JGen}}
\newcommand{\JVer}{\textbf{JVer}}
\newcommand{\Inst}{\textbf{InstGen}}
\newcommand{\Enc}{\textbf{Encode}}
\newcommand{\LR}{\textbf{LR}}
\newcommand{\prms}{\textbf{prms}}
\newcommand{\puzzle}{\textbf{puzzle}}
\newcommand{\Obf}{\textbf{Obfuscate}}
\newcommand{\Evaluate}{\textbf{Evaluate}}
\newcommand{\U}{U_{\lambda}}
\newcommand{\pk}{\mathrm{PK}_{FHE}}
\newcommand{\sk}{\mathrm{SK}_{FHE}}
\newcommand{\p}{\mathrm{Program}}
\newcommand{\RND}{\mathcal{RND}}
\newcommand{\GARBLE}{\textbf{GARBLE}}

\usepackage{algorithm2e}

\title{
 Indistinguishability Obfuscation of Circuits and its Application in Security\\
}

\author{
  Shilun Li \\
  Dept. of Computer Science \\
  Stanford University \\
  \texttt{shilun@stanford.edu} \\
  \And
  Zijing Di \\
  Dept. of Computer Science \\
  Stanford University \\
  \texttt{zidi@stanford.edu}\\
}

\begin{document}
\maketitle
\begin{abstract}
Under discussion in the paper is an $\iO$ (indistinguishability obfuscator) for circuits in Nick's Class. The obfuscator is constructed by encoding the Branching Program given by Barrington’s theorem using Multilinear Jigsaw Puzzle framework. We will show under various indistinguishability hardness assumptions, the constructed obfuscator is an $\iO$ for Nick's Class. Using Fully Homomorphic Encryption, we will amplify the result and construct an $\iO$ for $\Ppoly$, which are circuits of polynomial size. Discussion on $\iO$ and Functional Encryption is also included in this paper. 
\end{abstract}

\section{Introduction}
In 2001, Barak et al.~\cite{cryptoeprint:2001/069} proposed the notion of indistinguishability obfuscation and opens an entire new area of research. $\iO$ is a type of software obfuscation which obfuscates any two programs which have the same input-output behavior to be indistinguishable from each other. It is a very strong object such that almost every existing cryptographic primitives can be constructed using $\iO$. However, current constructions of $\iO$ are expensive and impractical to use, and research in this area is still active. We are going to discuss obfuscation and indistinguishability obfuscation of $\NCo$ circuits based on \cite{garg2013candidate}. This construction makes use of the multilinear maps assumption which is not very sound, but we think it gives us an good introduction on $\iO$ construction and techniques such as bootstrapping is useful to give us a $\Ppoly$ obfuscation from any arbitrary $\NCo$ obfuscation. In this report, We will describe the indistinguishability obfuscation of $\NCo$ and explain how to use it and Fully Homomorphic Encryption to achieve indistinguishability obfuscation for all circuits. For readers who are interested in $\iO$, they may also read the most recent construction of $\iO$ from four well-founded assumptions~\cite{10.1145/3406325.3451093}.

\section{Preliminaries}
\subsection{Indistinguishability Obfuscator}
Let us first start by defining indistinguishability obfuscator $\iO$ for circuit classes.
\begin{definition}[Indistinguishability Obfuscator]
An \textbf{indistinguishability obfuscator} $\iO$ for a circuit class $\left\{\mathcal{C}_{\lambda}\right\}$ is a uniform PPT (probabilistic polynomial-time) algorithm satisfying:
\begin{itemize}
    \item \textbf{Completeness}: For any security parameters $\lambda \in \mathbb{N}$, $C \in \mathcal{C}_{\lambda}$, and inputs $x$, we have
    $$
    \operatorname{Pr}\left[C^{\prime}(x)=C(x): C^{\prime} \leftarrow i \mathcal{O}(\lambda, C)\right]=1.
    $$
    \item \textbf{Indistinguishability}: For all security parameters $\lambda \in \mathbb{N}$, for all pairs of circuits $C_{0}, C_{1} \in \mathcal{C}_{\lambda}$ such that $C_{0}(x)=C_{1}(x)$ for all inputs $x$, $\iO(C_0)$ and $\iO(C_1)$ are computationally indistinguishable. In other words, for any PPT distinguisher $D$, there exists a negligible function $\epsilon$ (a function that grows slower than $1/p$ for any polynomial $p$) such that:
    $$
    \left|\Pr\left[D\left(i \mathcal{O}\left(\lambda, C_{0}\right)\right)\right]-\Pr\left[D\left(i \mathcal{O}\left(\lambda, C_{1}\right)\right)\right]\right| \leq \epsilon(\lambda).
    $$
\end{itemize}
\end{definition}

We can apply this definition naturally to the circuit class $\NCo$ and $\Ppoly$.

\begin{definition}[Indistinguishability Obfuscator for $\NCo$]
A uniform PPT algorithm $\iO$ is an \textbf{indistinguishability obfuscator  for $\NCo$} if the following holds: for all constants $c\in\mathbb{N}$, let $\mathcal{C}_\lambda$ be the class of circuits of depth at most $c\log\lambda$ and size at most $\lambda$, then $\iO(c,\cdot, \cdot)$ is an indistinguishability obfuscator for the class $\{\mathcal{C}_\lambda\}$.
\end{definition}

\begin{definition}[Indistinguishability Obfuscator for $\Ppoly$]
Let $\mathcal{C}_\lambda$ be the class of circuits of size at most $\lambda$. An \textbf{indistinguishability obfuscator for $\Ppoly$} is an indistinguishability obfuscator for the class $\{\mathcal{C}_\lambda\}$.
\end{definition}

\subsection{Oblivious Transfer}
An oblivious transfer (OT) is a protocol which allows the receiver to choose some out of all of the sender's inputs. It requires that at the ends of the protocol, the sender knows nothing about which inputs the receiver chooses, and the receiver knows nothing other than the inputs it chooses. For instance, in a 1-out-of-$n$ OT, the sender has $n$ messages $m_1, \cdots, m_n$ and the receiver has a index $i \in [n]$. At the end of OT, the sender learns nothing about $i$ and the receiver learns nothing other than $m_i$. OT is crucial in the obfuscation of the branching program for Bob to learn matrices corresponding to his input $y$ privately. We will mostly use OT as a blackbox in our construction of $\iO$.

\subsection{Multilinear Map}
\begin{definition}[a multilinear map]
Let $G_1, G_2, \cdots, G_k$ and $G_T$ be groups of prime order $p$ where $g_i \in G_i$ for $i \in [k]$ is a generator. A multilinear map is a polynomial time function $e: G_1 \times G_2 \times \cdots \times G_k \rightarrow G_T$ such that it is
\begin{itemize}
    \item multilinear: for $x_i \in G_i$, and $\alpha_i \in \Zp$ where $i \in [k]$, we have
    $$e(x_1^{\alpha_1}, \cdots, x_k^{\alpha_k}) = e(x_1, \cdots, x_k)^{\prod_{i} \alpha_i}.$$
    \item non-degenerate: $g_T = e(g_1, \cdots, g_k)$ is a generator of $G_T$.
\end{itemize}
\end{definition}

\subsection{Witness Indistinguishable Proof}
Witness Indistinguishable Proof (WIP) is a proof system for languages in NP. In WIP, the prover who has the witness will try to convince the verifier that $x \in L$ by sending verifier a proof. The verifier should learn nothing about the witness other than the fact that it exists. The "learn nothing" notion is framed as verifier having difficulty distinguishing different witnesses. A non-interactive WIP is WIP with no interaction between the prover and verifier, and it satisfies perfect soundness if it is impossible for the verifier to accept a proof when $x \notin L$.

\section{Multilinear Jigsaw Puzzles}
In this section, we will introduce a variant of multilinear maps called \textbf{Multilinear Jigsaw Puzzles}. They are similar to the GGH multilinear encoding schemes in \cite{garg2013candidate} except that only the party that generated the system parameters is able to encode elements in the Multilinear Jigsaw Puzzles scheme. There are two entities in the Multilinear Jigsaw Puzzle scheme: the \textbf{Jigsaw Generator}, and the \textbf{Jigsaw Verifier}. The Jigsaw Generator takes as input a description of the ``plaintext elements" and encodes the plaintext into jigsaw puzzle pieces. The name jigsaw puzzle pieces comes from the fact that they can only be meaningfully combined in very restricted ways, like a jigsaw puzzle. The Jigsaw Verifier takes as input the jigsaw puzzle pieces and a specific Multilinear Form for combining these pieces. The Jigsaw Verifier outputs 1 if jigsaw puzzle pieces and the Multilinear Form are ``successfully arranged" together.

While we hope to specify the input ``plaintext elements" to the Jigsaw Generator, but in our setting it is the generator itself that choose the plaintext space $\Zp$. So instead, we will use a \textbf{Jigsaw Specifier} which takes $p$ as input and outputs the ``plaintext elements" in $\Zp$ that the generator should encode and ``encoding levels" which are subsets of $[k]$ specifying the relative level of encoding.

\begin{definition}[Jigsaw Specifier]
A \textbf{Jigsaw Specifier} is a tuple $(k, \ell, A)$ where $k, \ell \in \mathbb{Z}^{+}$are parameters, and $A$ is a probabilistic circuit with the following behavior: On input a prime $p$, $A$ outputs the prime $p$ and an ordered set of $\ell$ pairs $\left(S_{1}, a_{1}\right),\left(S_{2}, a_{2}\right), \ldots,\left(S_{\ell}, a_{\ell}\right)$, where each $a_{i} \in \mathbb{Z}_{p}$ and each $S_{i} \subseteq[k]$.
\end{definition}

Now let us deine \textbf{Multilinear Forms}, which can be evaluated on the output of the Jigsaw Specifier. Multilinear Forms corresponds to arithmetic circuits with gates such as addition, negation, and subtraction as well as ignore gates.

\begin{definition}[Multilinear Form]
A \textbf{Multilinear Form} is a tuple $\mathcal{F}=(k, \ell, \Pi, F)$, with  parameters $k, \ell \in \mathbb{Z}^{+}$. $\Pi$ is a circuit with $\ell$ input wires, consisting of binary addition $\oplus$ gates, binary multiplication $\otimes$ gates, unary negation $\ominus$ gates, and unary ``ignore" $\square$ gates. $F$ is an assignment of an index set $I \subseteq[k]$ to every wire of $\Pi$. A Multilinear Form must satisfy the following constraints:
\begin{itemize}
    \item For every $\oplus$-gate or $\ominus$-gate, all the inputs and outputs of that gate are assigned the same set $I$.
    \item  For every $\otimes$-gate, its two inputs are assigned disjoint sets $I_{1}, I_{2} \subseteq[k]$, and its outputs are assigned the union set $I_{1} \cup I_{2}$.
    \item  The out-degree of all $\square$-gates is zero.
    \item  The output wire is assigned the set $[k]$.
\end{itemize}
Moreover, we call a Multilinear Form $(k, \ell, \Pi, F)$ \textbf{$\alpha$-bounded} if the size of the circuit $\Pi$ is at most $\alpha$.
\end{definition}

The first two constraints of the Multilinear Form is what makes it ``multilinear", in the sense that any index of $[k]$ can be added only once on any path from input to output. The following figure illustrates a multilinear form with input size $\ell=7$ and index set size $k=8$.
\begin{figure}[H]
    \centering
    \includegraphics[width=10cm]{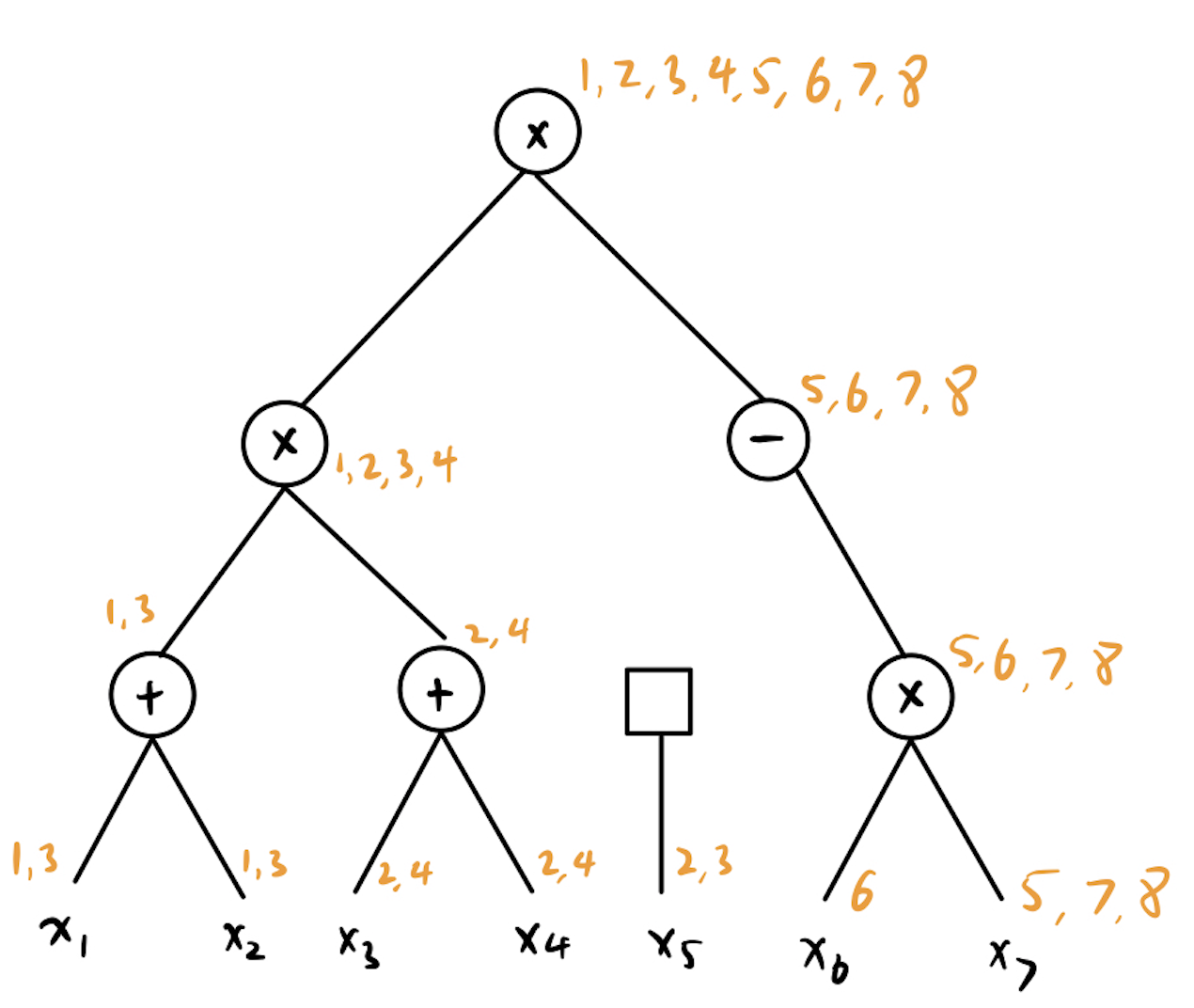}
    \caption{Example Multilinear Form}
    \label{fig: multilinear form}
\end{figure}

Now let us define what it means to ``evaluate" a Multilinear Form on the output of a Jigsaw Specifier.

\begin{definition}[Multilinear Evaluation]
For a Jigsaw Specifier with input $(k, \ell, A)$, let $X=\left(p,\left(S_{1}, a_{1}\right),\left(S_{2}, a_{2}\right), \ldots,\left(S_{\ell}, a_{\ell}\right)\right)$ denote its output, where $a_{i} \in \mathbb{Z}_{p}$  and $S_{i} \subseteq[k]$ for all $i$. We say that a Multilinear Form $\mathcal{F}=\left(k^{\prime}, \ell^{\prime}, \Pi, F\right)$ is \textbf{compatible with} $X$ if $k=k^{\prime}, \ell=\ell^{\prime}$, and the input wires of $\Pi$ are assigned the sets $S_{1}, S_{2}, \ldots, S_{\ell}$.

If the Multilinear Form $\mathcal{F}$ is compatible with $X$ then the \textbf{evaluation} $\mathcal{F}(X)$ is the output of the circuit $\Pi$ on the input $\left(S_{1}, a_{1}\right),\left(S_{2}, a_{2}\right), \ldots,\left(S_{\ell}, a_{\ell}\right)$ where the behavior of the gates are defined as follows (arithmetic operations are over $\mathbb{Z}_{p}$):
\begin{itemize}
    \item For every $\ominus$ gate, we have $\ominus(S, a)=(S,-a)$.
    \item For every $\oplus$ gate, we have $\left(S, a_{1}\right)\oplus\left(S, a_{2}\right)=\left(S, a_{1}+a_{2}\right)$.
    \item For every $\otimes$ gate, we have $\left(S_{1}, a_{1}\right)\otimes\left(S_{2}, a_{2}\right)=\left(S_{1} \cup S_{2}, a_{1} \cdot a_{2}\right)$.
\end{itemize}
We say that the multilinear evaluation $\mathcal{F}(X)$ \textbf{succeeds} if $\mathcal{F}(X)=([k], 0)$.
\end{definition}

Now we can formally define Multilinear Jigsaw Puzzles:

\begin{definition}[Multilinear Jigsaw Puzzle scheme]
A \textbf{Multilinear Jigsaw Puzzle scheme} $\mathcal{MJP}$ consists of two PPT algorithms $(\JGen, \JVer)$ defined as follows:

\textbf{Jigsaw Generator}: The generator $\JGen=(\Inst,\Enc)$ is specified via a pair of PPT algorithms:
\begin{itemize}
    \item $\Inst$ is the \textbf{randomized instance-generator} which takes as input the security parameter $1^{\lambda}$ and the multilinearity parameter $k$, and outputs a prime $p\leq 2^{\lambda}$, public system parameters $\prms$, and a secret state $s$ to pass to the encoding algorithm, $(p, \prms, s) \leftarrow \Inst\left(1^{\lambda}, 1^{k}\right)$.
    \item The (possibly randomized) \textbf{encoding algorithm} takes as input the prime $p$, the public parameters $\prms$ and secret state $s$, and a pair $(S, a)$ with $S \subseteq[k]$ and $a \in \mathbb{Z}_{p}$, and outputs an encoding of $a$ relative to $S$. We denote this encoding by $(S, u) \leftarrow \Enc(p, \prms, s, S, a)$.
\end{itemize}

The Jigsaw Generator is given a Jigsaw Specifier $(\ell, k, A)$ and the security parameter $\lambda$, it first runs the instance-generation to get $(p, \prms, s) \leftarrow \Inst\left(1^{\lambda}, 1^{k}\right)$, then runs the Jigsaw Specifier on input $p$ to get $X=\left(p,\left(S_{1}, a_{1}\right), \ldots,\left(S_{\ell}, a_{\ell}\right)\right) \leftarrow A(p)$, and finally encodes all the plaintext elements by running $\left(S_{i}, u_{i}\right) \leftarrow \Enc\left(\prms, s, S_{i}, a_{i}\right)$ for all $i=1, \ldots, \ell$.
The public output of the Jigsaw Generator, which we denote by $\puzzle=\left(\prms, \left(S_{1}, u_{1}\right), \ldots,\left(S_{\ell}, u_{\ell}\right)\right)$, consists of the parameters $\prms$, and also all the encodings $\left(S_{i}, u_{i}\right)$. For notational convenience, we denote the ``extended" output of $\JGen$ as
$$
(p, X, \puzzle) \leftarrow \JGen\left(1^{\lambda}, k, \ell, A\right).
$$
We call $\puzzle$ the public output and $X$ the private output. 

\textbf{Jigsaw Verifier}: The verifier $\JVer$ is a PPT algorithm that takes as input the public output $\puzzle$ of a Jigsaw Generator, and a Multilinear Form $\mathcal{F}=(k, \ell, \Pi, F)$. It outputs either 1 for accept or 0 for reject.

For a particular generator output $(p, X, \puzzle)$ and a form $\mathcal{F}$ compatible with $X$, we say that the verifier $\JVer$ is correct relative to $(p, \puzzle, \mathcal{F}, X)$ if 
$$\mathcal{F}(X)=([k], 0)\Longleftrightarrow \JVer(\puzzle, \mathcal{F})=1$$
Otherwise $\JVer$ is incorrect relative to $(p, \puzzle, \mathcal{F}, X)$.

We require that with high probability over the randomness of the generator, the verifier will be correct \textbf{on all forms}. Specifically, if $\JVer$ is deterministic then we require that for any polynomial $\alpha$ and $\alpha$-bounded Jigsaw Specifier family $\left\{\left(k_{\lambda}, \ell_{\lambda}, A_{\lambda}\right)\right\}_{\lambda \in \mathbb{Z}^{+}}$, the probability of $\JVer$ being incorrect relative to $(p, \puzzle, \mathcal{F}, X)$ for some $\alpha$-bounded form $\mathcal{F}$ compatible with $X$ is negligible, i.e.
$$
\Pr\left[(p, X, \puzzle) \leftarrow \JGen\left(1^{\lambda}, k_{\lambda}, \ell_{\lambda}, A_{\lambda}\right)\right] \leq \epsilon(\lambda) .
$$
for some negligible function $\epsilon$.
\end{definition}
The correctness follows directly from the definition of the Multilinear Jigsaw Puzzle framework. On the other hand, from the aspect of security, intuitively we want it to be the case that two different Jigsaw Puzzles $\puzzle_A$ and $\puzzle_{A'}$ are distinguishable \textbf{if and only if} there is a multilinear form $\mathcal{F}$ that succeeds with noticeably different probabilities on $A$ vs. $A'$. Thus, we will consider distributions over puzzles $\puzzle_A$ and $\puzzle_{A'}$ which are not distinguishable via multilinear forms. Under our assumptions, these puzzles are then computationally indistinguishable from each other.

Formally, the hardness assumption in the Multilinear Jigsaw Puzzle framework states that the public output of the Jigsaw Generator on two different polynomial-size families of Jigsaw Specifiers are computationally indistinguishable.



\section{Branching Program}
Our branching program is based on the one defined in the Barrington's theorem~\cite{barr1986branching}, called ``oblivious linear branching programs''. The theorem states that all $\log$-depth circuits ($\NCo$) have a corresponding poly-length branching program. Or in general, any depth-$d$ circuits have an equivalent branching program of length at most $4^d$. 
\begin{theorem}[\cite{barr1986branching}]
There exists two distinct 5-cycle permutation matrices $A_0, A_1 \in \zo^{5 \times 5}$ such that for any depth-d fan-in-2 Boolean circuit $C(\cdot)$, there exists an $(A_0, A_1)$ oblivious linear branching program of length at most $4^d$ that computes the same function as the circuit $C$.
\end{theorem}

As a result, if we can obfuscate a poly-length matrix branching program, then we prove that there exists a $\iO$ for $\NCo$.

The branching program computes a function $f$ using permutations in $S_5$ encoded by permutation matrices. In particular, multiplying permutations is equivalent to multiply matrices, each permutation corresponds to a matrix in $\zo^{5 \times 5}$ and the identity permutation corresponds to matrix $I$. To compute the functions, we are given two distinct matrices $A_0, A_1$ and break the computation into $i$ steps. In each step, for each bit $b$ of the input $x$, we multiply the matrix $A_{i, b} \in \zo^{5 \times 5}$. In the end, if the product of all matries is $A_0$, then $f(x) = 0$, else if the product is $A_1$, then $f(x) = 1$. In our later construction of $\iO$, we will let $A_0 = I$ and hence $f(x) = 0$ if the product of all matrices is the identity matrix. 
\begin{definition}[Oblivious Linear Branching Program]
Given two distinct permutation matrices $A_0, A_1 \in \zo^{5 \times 5}$, an $(A_0, A_1)$ oblivious branching program of length-$n$ for a $l$-bit input is the sequence
$$BP = ((\textbf{inp}(i), A_{i,0}, A_{i,1}))^n_{i=1},$$
where $A_{i,0} \in \zo^{5 \times 5}$ is the permutation matrix corresponds to bit 0 of $\textbf{inp}(i)$-th bit of $x$ in step $i$, $A_{i,1} \in \zo^{5 \times 5}$ corresponds to bit 1 of $\textbf{inp}(i)$-th bit of $x$. The function $f$ is computed as 
\[ f_{BP, A_O, A_1}(x) = \begin{cases} 
      0 & \prod_{i=1}^n A_{i, x_{\textbf{inp}(i)}} = A_0 \\
      1 & \prod_{i=1}^n A_{i, x_{\textbf{inp}(i)}} = A_1 \\
      \textbf{undef} & \text{otherwise}
   \end{cases}
\]
\end{definition}
For a concrete example, suppose $f(x_1x_2) = x_1 \oplus x_2$. Then we can define 
\[A_0 = I,
A_1 = \begin{pmatrix}
0 & 1\\
1 & 0
\end{pmatrix},
A_{1,0} = \begin{pmatrix}
1 & 0\\
0 & 1
\end{pmatrix},
A_{1,1} = \begin{pmatrix}
0 & 1\\
1 & 0
\end{pmatrix},
A_{2,0} = \begin{pmatrix}
1 & 0\\
0 & 1
\end{pmatrix},
A_{2,1} = \begin{pmatrix}
0 & 1\\
1 & 0
\end{pmatrix}.\]
Indeed, if $x=00$, then at step 1, we multiply by $A_{1,0}$ and at step 2 multiply by $A_{2,0}$. $f$ correctly outputs 0 since $A_{1,0}A_{2,0} = I$. Other three cases also easily follow.

\section{\texorpdfstring{$\NCo$}{TEXT} Obfuscation}
In this section, We will describe the indistinguishability obfuscation candidate for $\NCo$ which uses branching problems and Multilinear Jigsaw Puzzles. Let us denote a length $n$ oblivious branching program over $\ell$ inputs by
$$
B P=\left\{\left(\textbf{inp}(i), A_{i, 0}, A_{i, 1}\right): i \in[n], \textbf{inp}(i) \in[\ell], A_{i, b} \in\{0,1\}^{5 \times 5}\right\}.
$$
\subsection{Randomized Branching Program}
Let us randomize the branching program over some ring $\Zp$. Let $m=2n+5$ and the randomized branching program can be described as follows:
\begin{itemize}
\item Randomly sample $\left\{\alpha_{i, 0}, \alpha_{i, 1}, \alpha_{i, 0}^{\prime}, \alpha_{i, 1}^{\prime}: i \in[n]\right\}$ from $\Zp$, subject to the constraint that $\prod_{i \in I_{j}} \alpha_{i, 0}=\prod_{i \in I_{j}} \alpha_{i, 0}^{\prime}$ and $\prod_{i \in I_{j}} \alpha_{i, 1}=\prod_{i \in I_{j}} \alpha_{i, 1}^{\prime}$ for all $j \in[\ell]$.

\item For every $i \in[n]$, compute four $(2 m+5) \times(2 m+5)$ block-diagonal matrices $D_{i, 0}, D_{i, 1}, D_{i, 0}^{\prime}, D_{i, 1}^{\prime}$ where the diagonal entries $1, \ldots, 2 m$ are chosen at random. The bottom-right $5 \times 5$ block of $D_{i,b}$ is given by $\alpha_{i,b} A_{j, b}$ and the bottom-right $5 \times 5$ block of $D_{i,b}$ is given by $\alpha_{i,b}^\prime I_{5}$, for $b=0,1$.

\item Choose vectors $\mathbf{s}$ and $\mathbf{t}$, and $\mathbf{s}^{\prime}$ and $\mathbf{t}^{\prime}$ of dimension $2 m+5$ as follows:
\begin{align*}
    &\mathbf{s}=(0_m, v_m, \mathbf{s}^*)  &\mathbf{s}^\prime=(0_m, v_m^\prime, \mathbf{s}^{\prime*})\\
    &\mathbf{t}=(0_m, w_m, \mathbf{w}^*)^T &\mathbf{t}^\prime =(0_m, w_m^\prime, \mathbf{t}^{\prime*})^T
\end{align*}
where $v_m,v^\prime_m,w_m,w^\prime_m$ are random vectors in $\Zp^m$ and $\mathbf{s}^*, \mathbf{s}^{\prime*}, \mathbf{t}^*, \mathbf{t}^{\prime*}$ are are random vectors in $\Zp^5$ subject to the constraint that $\langle \mathbf{s}^*, \mathbf{t}^*\rangle=\langle \mathbf{s}^{\prime*}, \mathbf{t}^{\prime*}\rangle$.

\item Sample random full-rank $(2 m+5) \times(2 m+5)$ matrices $R_{0}, R_{1}, \ldots, R_{n}$ and $R_{0}^{\prime}, R_{1}^{\prime}, \ldots, R_{n}^{\prime}$ over $\Zp$ and compute their inverses.

\item The randomized branching program over $\mathbb{Z}_{p}$ is the following:
$$
\begin{aligned}
&\RND_{p}(B P)= \\
&\left\{\begin{array}{ll}
\tilde{\mathbf{s}}=\mathbf{s} R_{0}^{-1}, \tilde{\mathbf{t}}=R_{n} \mathbf{t}, & \tilde{\mathbf{s}}^{\prime}=\mathbf{s}^{\prime}\left(R_{0}^{\prime}\right)^{-1}, \tilde{\mathbf{t}}^{\prime}=R_{n}^{\prime} \mathbf{t}^{\prime} \\
\left\{\tilde{D}_{i, b}=R_{i-1} D_{i, b} R_{i}^{-1}: i \in[n], b \in\{0,1\}\right\}, & \left\{\tilde{D}_{i, b}^{\prime}=R_{i-1}^{\prime} D_{i, b}^{\prime}\left(R_{i}^{\prime}\right)^{-1}: i \in[n], b \in\{0,1\}\right\}
\end{array}\right\}
\end{aligned}
$$
\end{itemize}
The randomized branching program runs the original branching program $BP$ and a ``dummy program" consisting only of identity matrices at the same time. We only use the dummy program for the purpose of equality test, i.e. the original program outputs 1 only when it agrees with the dummy program.

\subsection{Garbled Branching Program}
Now using Multilinear Jigsaw Puzzles, we can use the Jigsaw Specifier that on input $p$ randomizes the branching program over $\Zp$ and outputs $\RND_p(BP)$. Then we use the encoding part of the Jigsaw generator to encode each element of the step-$i$ matrices relative to $\{i+1\}$, each element of the vectors $\tilde{\mathbf{s}}, \tilde{\mathbf{s}}^{\prime}$ relative to $\{1\}$, and each element of the vectors $\tilde{\mathbf{t}}, \tilde{\mathbf{t}}^{\prime}$ relative to $\{n+2\}$. We denote the public output of the Jigsaw generator, which is call the randomized and encoded program, by
$$
\begin{aligned}
&\widehat{\RND}_{p}(B P)= \\
&\left\{\begin{array}{ll}
\prms, \quad \hat{\mathbf{s}}=\Enc_{\{1\}}(\tilde{\mathbf{s}}), \hat{\mathbf{t}}=\Enc_{\{n+2\}}(\tilde{\mathbf{t}}), & \hat{\mathbf{s}}^{\prime}=\Enc_{\{1\}}\left(\tilde{\mathbf{s}}^{\prime}\right), \hat{\mathbf{t}}^{\prime}=\Enc_{\{n+2\}}\left(\tilde{\mathbf{t}}^{\prime}\right) \\
\left\{\hat{D}_{i, b}=\Enc_{\{i+1\}}\left(\tilde{D}_{i, b}\right): i \in[n], b \in\{0,1\}\right\}, & \left\{\hat{D}_{i, b}^{\prime}=\Enc_{\{i+1\}}\left(\tilde{D}_{i, b}^{\prime}\right): i \in[n], b \in\{0,1\}\right\}
\end{array}\right\}
\end{aligned}
$$
On the other hand, the private output of the Jigsaw generator is $(p,\RND_p(BP))$. 

For every input $\chi\in\{0,1\}^{\ell}$ to $BP$, let us consider the Multilinear Form $\mathcal{F}_{\chi}$ given by
$$
\mathcal{F}_{\chi}\left(\RND_{p}(B P)\right)=\tilde{\mathbf{s}}\left(\prod_{i} \tilde{D}_{i, \chi_{\textbf{inp}(i)}}\right) \tilde{\mathbf{t}}-\tilde{\mathbf{s}}^{\prime}\left(\prod_{i} \tilde{D}_{i, \chi_{\textbf{inp}(i)}^{\prime}}\right) \tilde{\mathbf{t}}^{\prime} \bmod p
$$
Note that 
\begin{align*}
    BP(\chi)&=0\quad\Rightarrow \quad \Pr \left[\mathcal{F}_{\chi}\left(\RND_{p}(BP)\right)=0\right]=1,\\
    BP(\chi)&=1\quad\Rightarrow \quad \Pr\left[\mathcal{F}_{\chi}\left(\RND_{p}(BP)\right)=0\right]=1/p.
\end{align*}
So given the public output $\widehat{\RND}_{p}(B P)$ of the generator and the original input $\chi$, we can use the Jigsaw verifier to check if $\mathcal{F}_{\chi}\left(\RND_{p}(B P)\right)=0$ and learn the output of $B P(\chi)$ with high probability.

Roughly speaking, we want it to be the case where if for two different ways of fixing inputs to the branching program result in the same functionality, then it is infeasible to decide which of the two sets of fixed inputs is used in a given garbled program.

Given $\widehat{\RND}_{p}(B P)$ and a partial assignment for the input bits, $\sigma: J \rightarrow\{0,1\}$ for some $J \subset[\ell]$, the Parameter-fixing procedure removes all the matrices $\tilde{D}_{i, b}, \tilde{D}_{i, b}^{\prime}$ that are not consistent with that partial assignment $\sigma$. So we have
\begin{align*}
&\GARBLE\left(\widehat{\mathcal{R N D}}_{p}(B P),(J, \sigma)\right)=\\
&\left\{\begin{array}{ll}
\prms, \hat{\mathbf{s}}, \hat{\mathbf{t}}, & \hat{\mathbf{s}}^{\prime}, \hat{\mathbf{t}}^{\prime} \\
\left\{\hat{D}_{i, b}: i \in I_{J}, b=\sigma(\textbf{inp}(i))\right\}, & \left\{\hat{D}_{i, b}^{\prime}: i \in I_{J}, b=\sigma(\textbf{inp}(i))\right\} \\
\left\{\hat{D}_{i, b}: i \notin I_{J}, b \in\{0,1\}\right\}, & \left\{\hat{D}_{i, b}^{\prime}: i \notin I_{J}, b \in\{0,1\}\right\}
\end{array}\right\}
\end{align*}
The garbled program can be thought as input fixing. If the underlying program is computing a function $F$ then the garbled program computes $\left.F\right|_{\sigma}$.
\begin{definition}[Functionally Equivalent Assignments]
Fix a function $F:\{0,1\}^{\ell} \rightarrow\{0,1\}$. Two partial assignments $\left(J, \sigma_{0}\right)$, $\left(J, \sigma_{1}\right)$ over the same input variables are \textbf{functionally equivalent} relative to $F$ if $\left.F\right|_{\sigma_{0}}=\left.F\right|_{\sigma_{1}}$.
\end{definition}
\begin{assumption}[Equivalent Program Indistinguishability]
\label{assump: equivalent program}
For any length $\ell$ branching program $BP$ computing a function $F$, and any two functionally equivalent partial assignments relative to $F$, $\left(J, \sigma_{0}\right)$ and $\left(J, \sigma_{1}\right)$, the corresponding garbled programs are computationally indistinguishable:
$$
\GARBLE\left(\widehat{\RND}(B P),\left(J, \sigma_{0}\right)\right) \approx \GARBLE\left(\widehat{\RND}(B P),\left(J, \sigma_{1}\right)\right).
$$
\end{assumption}
\subsection{Candidate \texorpdfstring{$\NCo$}{TEXT} Obfuscator}
Our candidate $\NCo$ obfuscator will be of the form $\GARBLE\left(\widehat{\RND}(B P),\left(J, \sigma_{1}\right)\right)$:
\begin{theorem}
Under the assumptions given by \Cref{assump: equivalent program}, there exists an efficient $i \mathcal{O}$ for $\NCo$ circuits.
\end{theorem}
\begin{proof}
Fix a constant $\gamma$, for any security parameter $\lambda$ let $\mathcal{C}_{\lambda}$ be the class of circuits of depth $\gamma \log \lambda$ and size at most $\lambda$. Let $U_{\lambda}$ be a poly-sized universal circuit for this circuit class, i.e. $U_{\lambda}(C, m)=C(m)$ for all $C \in \mathcal{C}_{\lambda}$ and $m \in\{0,1\}^{n}$. Furthermore, all circuits $C \in \mathcal{C}_{\lambda}$ can be encoded as an $\ell=\ell(\lambda)$ bit string as input to $U$. Let $UBP_{\lambda}(C, m)$ be the branching program of the universal circuit $U_{\lambda}(C, m)$ obtained by applying Barrington's theorem.

Denote by $I_{C}$ the steps in the program $UBP_{\lambda}$ that examine the input bits from the $C$ input, and for each particular circuit $c$ denote by $\left(I_{C}, \sigma_{c}\right)$ the partial assignment that fixes the bits of that circuit in the input of $UBP_{\lambda}$. The obfuscator is then given by:
$$
\iO(\lambda, c)=\GARBLE\left(\widehat{\RND}\left(U B P_{\lambda}\right),\left(I_{C}, \sigma_{c}\right)\right).
$$
Functionality and polynomial slowdown are obviou\Cref{assump: equivalent program} directly asserts that for any two circuits $c_{1}, c_{2}$ that compute the same function, $\operatorname{UBP}\left(c_{1}, \cdot\right) = U B P\left(c_{2}, \cdot\right)$ and $\iO\left(\lambda, c_{1}\right)$ is computationally indistinguishable from $\iO\left(\lambda, c_{2}\right)$.
\end{proof}

\section{Fully Homomorphic Encryption}
Aside from the usage of Fully Homomorphic Encryption (FHE) in $\iO$, FHE itself is a very powerful object and was awarded the G\"{o}del Prize in 2022. It is basically an encryption scheme enabled with the ability to apply any function on the ciphertexts. 
\begin{definition}
Let $\ckt$ be a class of all polynomial sized circuits, a Fully Homomorphic Encryption scheme is an encryption scheme $(\KeyG, \Encrypt, \Dec, \Eval)$ where
\begin{itemize}
    \item $(\KeyG, \Encrypt, \Dec)$ are common to encryption scheme. In public-key encryption, for instance,
    \begin{itemize}
        \item $\KeyG(1^{\lambda})\to (pk, sk)$: $\KeyG$ takes in the security parameter $\lambda$ and outputs the public key, secret key pair;
        \item $\Encrypt(pk, m) \tor ct$: $\Encrypt$ encrypts a message with the public key $pk$ and generates a ciphertext $ct$ ($\tor$ indicates that $\Encrypt$ might be randomized);
        \item $\Dec(sk, ct) \to m$: $\Dec$ decrypts the ciphertext into the original message using the secret key.
    \end{itemize}
    \item $\Eval$ is the additional function which enables functions to apply to ciphertexts.
    
    $\Eval(pk, f, ct_1, \cdots, ct_n) \to ct^*$: $\Eval$ takes in a function $f \in \ckt$ and $n$ ciphertexts where $ct_i = \Encrypt(pk, m_i)$, and returns a new ciphertext $ct^*$ where $ct^* = \Encrypt(pk, f(m_1, \cdots, m_n))$. In English, $\Eval$ gives back the ciphertext of the output of $f$ applied to $n$ messages given only $n$ ciphertexts of the messages. 
\end{itemize}
\end{definition}

A major application scenario would be computation outsourcing in an untrusted setup. Suppose you do not have enough computational resources to compute function $f$ on your secret data $(m_1, \cdots, m_n)$, then with FHE, you can send the ciphertexts $ct_1, \cdots, ct_n$ instead. The cloud will run $\Eval$ with the public function $f$ and the ciphertexts. The cloud sends back you the resulting new ciphertext $ct^* = \Encrypt(pk, f(m_1, \cdots, m_n))$. You who has the secret key $sk$ can decrypt $ct^*$ and get $f(m_1, \cdots, m_n)$ as desired.
\subsection{Security Definition of FHE}
The security notion for FHE is IND-CPA (or other variants such as IND-CCA) as for other encryption schemes. Compactness requirement is left for interested readers to explore more. The definition is often modeled as a game played between a challenger and an adversary. The challenger runs the FHE scheme honestly and provides an $\LR$ oracle for an adversary to query. The IND-CPA game $G$ is defined as follows:
\begin{itemize}
    \item The challenger samples a random bit $b \in \zo$ and generates keys from $\KeyG(1^{\lambda})$.
    \item The adversary chooses two messages $m_0, m_1$, and queries the $\LR$ oracle to get back $ct \getsr \Encrypt(m_b)$. The adversary can repeat this step multiple times (within their computational limits).
    \item The adversary outputs $b' \in \zo$.
    \item The game outputs 1 if $b = b'$. and 0 otherwise.
\end{itemize}
The advantage of an adversary $A$ is $Adv^{ind-cpa}(A) = |\Pr[G \Rightarrow 1]-\frac{1}{2}|$. If the advantage is small, it means that the adversary cannot do better than guessing a random bit. 
\begin{definition}
A (fully homomorphic) encryption scheme is IND-CPA secure if for any PPT adversary $A$, $Adv^{ind-cpa}(A)$ is negligible in the security parameter $\lambda$.
\end{definition}
The definition makes sense since if the adversary cannot distinguish ciphertexts of any two messages of their own choice, then they cannot obtain any useful information about the message encrypted. Notice that $\Eval$ is not modeled as an oracle. That's because $\Eval$ is completely public and adversary can run it on their own, and the security should still hold.

\section{\texorpdfstring{$\Ppoly$}{TEXT} Obfuscation}
We can now construct a candidiate obfuscation for $\Ppoly$ circuits given the $\NCo$ obfuscation $\iO_{\NCo}$, WIP and a FHE scheme $(\KeyG_{FHE}, \Encrypt_{FHE}, \Dec_{FHE}, \Eval_{FHE})$. Let $\{\ckt_{\lambda}\}$ be a family of circuit classes with input size and circuit size polynomial in $\lambda$. Let $\{\U\}$ be a family of poly-sized universal circuits such that $\U(C,m) = C(m)$ for all $C \in \ckt_{\lambda}$ and all $m \in \zo^{\mathrm{poly}(\lambda)}$. $C$ as an input for $\U$ can be encoded as a $\ell=\mathrm{poly}(\lambda)$ bit string. Also let $\U(\cdot, m)$ be the universal circuit hardwired with $m$ that only has $C$ as the input. The construction consists of $\Obf$ and $\Evaluate$ algorithms:
\RestyleAlgo{boxruled}
\LinesNumbered
\begin{algorithm}[ht]
  \caption{$\Obf(1^{\lambda}, C \in \ckt_{\lambda})$}
  $(\pk^1, \sk^1) \getsr \KeyG_{FHE}(1^{\lambda})$; $(\pk^2, \sk^2) \getsr \KeyG_{FHE}(1^{\lambda})$\\
  $g_1 \getsr \Encrypt_{FHE}(\pk^1, C)$; $g_2 \getsr \Encrypt_{FHE}(\pk^2, C)$\\
  $P \getsr \iO_{\NCo}(\p1^{\sk^1, g_1, g_2})$\\
  Return $\sigma = (P, \pk^1, \pk^2, g_1, g_2)$
\end{algorithm}
\begin{algorithm}[ht]
  \caption{$\Evaluate(\sigma, m)$}
  $e_1 \getsr \Eval_{FHE}(\pk^1, \U(\cdot, m), g_1)$; $e_2 \getsr \Eval_{FHE}(\pk^2, \U(\cdot, m), g_2)$\\
  Simulate the prover from WIP and generates a proof $\phi$ to prove that $e_1$ and $e_2$ are computed correctly.\\
  Return $P(m, e^1, e^2, \phi)$
\end{algorithm}

$\p1$ and $\p2$ are two algorithms to verify the proof $\phi$ and output the corresponding decryption result if the verification passes. $\p2$ has exactly the same codes as $\p1$ except the place where red text is placed. Notice that $\phi$ is a proof that can be verified by a $\log$ path circuit in $\NCo$, and this is why we can apply $\iO_{\NCo}$ to the programs. The perfect soundness of WIP ensures that it is impossible to generate a fake proof $\phi$ (will output 0 one hundred percent of the time if $e_1$ and $e_2$ are not generated correctly), and the indistinguishability property ensures that the witness for $\p1$ and $\p2$ are indistinguishable.
\begin{algorithm}[ht]
  \caption{$\p1^{\sk^1, g_1, g_2}(m, e_1, e_2, \phi)$\textcolor{red}{//$\p2^{\sk^2, g_1, g_2}(m, e_1, e_2, \phi)$}}
  Verify if $\phi$ is a valid proof for the NP statement ($e_1 = \Eval_{FHE}(\pk^1, \U(\cdot, m), g_1)$ and $e_2 = \Eval_{FHE}(\pk^2, \U(\cdot, m), g_2)$) \\
  Return $\Dec_{FHE}(\sk^1, e_1)$ if $\phi$ is valid; else return 0 \textcolor{red}{//Return $\Dec_{FHE}(\sk^2, e_2)$ if check passes instead}
\end{algorithm}

The rough idea is that we use FHE to hide the circuit as a ciphertext and apply circuit to the ciphertext using $\Eval$, and use $\iO_{\NCo}$ to verify FHE evaluation is computed correctly. Following the correctness of FHE and $\iO_{\NCo}$, if everything is honestly generated, then $g_1$ and $g_2$ are the ciphertexts of $C$ under two different public keys, $P$ has the same input-output behavior as $\p1$. Then $e_1$ is the ciphertext of the function $\U(\cdot, m)$ applied to the original message $C$ of the ciphertext $g_1$, i.e., $e_1 = \Encrypt_{FHE}(\pk^1, \U(C, m)) = \Encrypt_{FHE}(\pk^1, C(m))$. Similarly, $e_2 = \Encrypt_{FHE}(\pk^2, C(m))$. Finally, $\p1$ will pass the check and output $\Dec_{FHE}(\sk^1, e_1) = \Dec_{FHE}(\sk^1, \Encrypt_{FHE}(\pk^2, C(m))) = C(m)$, which is also the output of $P$ because they have the same behavior. Hence, at the end of $\Evaluate$, the obfuscation indeed returns $C(m)$ as desired.

Intuitively, given $\sigma$, this construction obfuscates $C$ since FHE satisfies IND-CPA security, and an adversary cannot get any information about $C$ from the ciphertexts. The adversary cannot get information from $P$ neither since we assume $\iO_{\NCo}$ satisfies indistinguishability.

\begin{theorem}
The $\iO$ construction above for all poly-sized circuits is secure in the indistinguishability game under assumptions we made earlier. 
\end{theorem}
\begin{proof}
We proceed the proof by a series of hybrid arguments. Recall that in the indistinguishability game, we want to show $\iO(C_0)$ and $\iO(C_1)$ are computationally indistinguishable for all pairs of $C_0, C_1 \in \ckt_{\lambda}$ where $C_0(x) = C_1(x)$ for all inputs $x$. In hybrid arguments, we start from the game for $\iO(C_0)$ and move one step a time until we get to the game for $\iO(C_1)$, and each move should be computationally indistinguishable from each other.
\begin{itemize}
    \item Hyb$_0$: We start from the honestly executed scheme for $C_0$. It is the construction defined above with $C=C_0$. In particular, $g_1 = \Encrypt_{FHE}(\pk^1, C_0)$, $g_2 = \Encrypt_{FHE}(\pk^2, C_0)$ and $P = \iO_{\NCo}(\p1^{\sk^1, g_1, g_2})$. All other codes stay the same. The ultimate goal is to move to the scheme for $C=C_1$.
    \item Hyb$_1$: Replace the ciphertext $g_2 = \Encrypt_{FHE}(\pk^2, C_1)$.
    \item Hyb$_2$: Replace the obfuscated program $P = \iO_{\NCo}(\p2^{\sk^2, g_1, g_2})$.
    \item Hyb$_3$: Replace the ciphertext $g_1 = \Encrypt_{FHE}(\pk^1, C_1)$.
    \item Hyb$_4$: Replace the obfuscated program $P = \iO_{\NCo}(\p1^{\sk^1, g_1, g_2})$. In particular, we now have $g_1 = \Encrypt_{FHE}(\pk^1, C_1)$, $g_2 = \Encrypt_{FHE}(\pk^2, C_1)$ and $P = \iO_{\NCo}(\p1^{\sk^1, g_1, g_2})$.
\end{itemize}

Next, we want to prove that each hybrid is computationally indistinguishable from the next one, and hence Hyb$_0$ is indistinguishable from Hyb$_4$.
\begin{itemize}
    \item Hyb$_0$ to Hyb$_1$: Hyb$_1$ basically changes from the ciphertext of $C_0$ to the ciphertext of $C_1$ under the same key $\pk^2$. Suppose by contradiction, Hyb$_0$ is not indistinguishable from Hyb$_1$, then there exists a distinguisher. We can then construct an IND-CPA adversary which queries the challenger with two messages $C_0$ and $C_1$, and receive a ciphertext $\Encrypt_{FHE}(\pk^2, C_b)$. Then the IND-CPA adversary can use the ciphertext to simulate the $\iO$ construction and ask the distinguisher to distinguish which ciphertext is used. Since we successfully construct an IND-CPA adversary for the FHE scheme, we reach the contradiction that the FHE is IND-CPA secure.
    \item Hyb$_1$ to Hyb$_2$: In this move, we change $P$ from the obfuscation of $\p1$ to the obfuscation of $\p2$. Notice that $\p1(x) = \p2(x)$ for all inputs: $\phi$ is the proof for the same $e_1$ and $e_2$, $\Dec_{FHE}(\sk^1, e_1) = C_0(m) = C_1(m) = \Dec_{FHE}(\sk^2, e_2)$. Therefore, if we can distinguish Hyb$_1$ from Hyb$_2$, then we can distinguish $\iO(\p1)$ from $\iO(\p2)$ where $\p1$ and $\p2$ can be computed by a $\NCo$ circuit, contradicting the assumption that our $\iO_{\NCo}$ satisfies indistinguishability.
    \item Hyb$_2$ to Hyb$_3$: Follows the same reasoning as Hyb$_0$ to Hyb$_1$.
    \item Hyb$_3$ to Hyb$_4$: Follows the same reasoning as Hyb$_1$ to Hyb$_2$.
\end{itemize}
\end{proof}
\pagebreak
\bibliography{main}
\bibliographystyle{alpha}

\end{document}